\documentclass[10pt]{article}
\usepackage{amsmath}
\usepackage{amsfonts}
\usepackage{amssymb}
\usepackage{amsthm}
\usepackage{caption}
\usepackage{subcaption}
\usepackage[ruled,linesnumbered,vlined]{algorithm2e}
\usepackage{tikz}
\usepackage{float}
\usepackage{authblk}
\usetikzlibrary{graphs}
\theoremstyle{definition}

\newtheorem{theorem}{Theorem}[section]

\author[1]{Isabel Beichl}
\affil[1]{National Institute of Standards and Technology}
\author[2]{Alathea Jensen}
\affil[2]{George Mason University}
\date{August 1, 2017}
\begin{document}
\title{A Sequential Importance Sampling Algorithm for Estimating Linear Extensions}
\maketitle
\begin{abstract}
In recent decades, a number of profound theorems concerning approximation of hard counting problems have appeared. These include estimation of the permanent, estimating the volume of a convex polyhedron, and counting (approximately) the number of linear extensions of a partially ordered set.  All of these results have been achieved using probabilistic sampling methods, specifically Monte Carlo Markov Chain (MCMC) techniques.  In each case, a rapidly mixing Markov chain is defined that is guaranteed to produce, with high probability, an accurate result after only a polynomial number of operations. 

Although of polynomial complexity, none of these results lead to a practical computational technique, nor do they claim to. The polynomials are of high degree and a non-trivial amount of computing is required to get even a single sample.  Our aim in this paper is to present practical Monte Carlo methods for one of these problems, counting linear extensions.  Like related work on estimating the coefficients of the reliability polynomial, our technique is based on improving the so-called Knuth counting algorithm by incorporating an importance function into the node selection technique giving a sequential importance sampling (SIS) method.  We define and report performance on two importance functions.
\end{abstract}
\section{Introduction}
For a partially ordered set, any total order which respects the partial order is known as a linear extension.   Determining the number of linear extensions of a given partially ordered set (poset) is a fundamental problem in the study of ordering with many applications.

The most common applications are in scheduling problems, where the number of linear extensions gives the size of the solution space.  For example, if we are trying to find a schedule that is optimal with respect to some cost function, the size of the solution space can be valuable in deciding how to carry out the search.  If the solution space is small, then we may be able to perform an exhaustive search for the best order, but if the solution space is large, then we may wish to settle for an approximate solution.

Computing the number of linear extensions exactly is $\#$P-complete \cite{BW}, but there are fully polynomial-time randomized approximation schemes to estimate the number, such as the Markov Chain Monte Carlo (MCMC) method presented in \cite{BW}.  These algorithms, while of great theoretical importance, are not practical because the polynomials are of high degree, even for the improved versions due to Bubbly-Dyer \cite{bubdye} and Wilson \cite{wilson}.

Our aim is a practical method based on sequential importance sampling (SIS) in the spirit of the applications discussed in \cite{jcp} and \cite{algo}. Our method may have exponential rather than polynomial complexity in some cases, but, as is typical of SIS, a modest amount of computation is often sufficient to learn at least something plausible about the number of linear extensions for any poset of interest.  This is not the case for MCMC methods where significant computation may be required to obtain even a single sample chosen from a probability distribution close to the limit distribution.

In Section \ref{Preliminaries}, we give definitions and notation necessary to describe our algorithm.  In Section \ref{Algorithm}, we give the algorithm and explain how it works.  The algorithm takes an importance function as an input, and in Section \ref{ImportanceFunction} we discuss the relative merits of several importance functions.  In Section \ref{Variance}, we give an expression for the variance of the estimates produced by the algorithm and discuss methods for bounding the variance.  In Section \ref{Improvements}, we discuss possible improvements to the algorithm, and in Section \ref{NumericalResults} we give results of our numerical experiments with randomly generated posets.

\section{Preliminaries} \label{Preliminaries}

For a partially ordered set (poset), $(P,\leq)$, we use the notation $x\leq y$ to mean that $x$ \emph{precedes} $y$ in the partial order or that $x=y$.  If $x$ and $y$ are not equal but $x\leq y$, we write $x<y$.  If $x<y$ and there is no $z$ such that $x<z<y$, we say that $y$ \emph{covers} $x$.

By the \emph{descendants} of an element $v\in P$ we mean all $w\in P$ such that $w\leq v$.  We will refer to the number of descendants of $v\in P$ as $d(v)$.  By the \emph{ancestors} of $v\in P$ we mean all $w\in P$ such that $v\leq w$.

A \emph{linear extension} of $P$ is any permutation $\sigma:P\to \{1,2,\dots,|P|\}$ such that for all $x,y\in P$, $x\leq y$ implies that $\sigma(x)\leq\sigma(y)$.  We denote the set of all linear extensions of a poset $P$ by $\Lambda(P)$, and the number of linear extensions by $|\Lambda(P)|$. If $P$ has $n$ elements, then $|\Lambda(P)|$ is some integer between $1$ and $n!$.

In practice, posets are often arise from directed acyclic graphs (DAGs), and in this context, a linear extension is usually called a \emph{topological ordering} or \emph{topological sort}.  Each DAG has a unique poset that corresponds to its reachability relation; that is, $w\leq v$ in the poset if and only if there is a directed walk from $v$ to $w$ in the DAG.

Determining the reachability poset of a DAG amounts to taking the transitive reduction or the transitive closure of the DAG.  This can be accomplished by boolean matrix multiplication and hence has the same complexity as that procedure \cite{AGU}.

In general, different DAGs can have the same reachability relation, but we can create a unique DAG from a poset by directing an edge from $v$ to $w$ if and only if $v$ covers $w$ in the poset.

\section{The Algorithm} \label{Algorithm}

Generating a single linear extension of a poset is well-understood and can be done in time $O(n)$ for a poset with $n$ elements. For example, any depth-first or breadth-first search of a DAG will generate a linear extension as a side-effect.  For a breadth-first search, we can order the elements by their first visit times in ascending order, and for a depth-first search, we can order the elements by their last visit times in descending order.  A greedy search will also generate a linear extension.

The standard method of obtaining a linear extension, which was first described in a 1962 paper by Kahn \cite{kahn}, corresponds to a breadth-first search.  The procedure is as follows.  We choose some minimal element of the poset and delete it, possibly causing some new elements to become minimal.  We choose another minimal element of the poset and delete it.  We repeat this until there are no elements left in the poset.  The order in which the elements were deleted is then a linear extension of the poset.

We can also perform the same procedure with maximal elements instead of minimal, and in this case the reverse order of the deletions gives the linear extension.  In this paper, we will use maximal elements instead of minimal.

Since all linear extensions can be obtained by sequential deletion of maximal elements, $|\Lambda(P)|$ is equal to the number of ways to execute this procedure.  We can think of the choices available to us at each step as forming a decision tree, where each branching corresponds to a choice of a maximal element, and each leaf corresponds to one complete linear extension.  Hence the number of leaves is $|\Lambda(P)|$.

In the classic Knuth algorithm \cite{knuth} for estimating the number of leaves of a tree, we walk a path from the root to a leaf of the decision tree by choosing which branch to take at each level uniformly at random from the branches available.  We then take the product of the number of branches available at each step along the path to be an estimate of the number of leaves in the tree.  Knuth showed that this estimate is unbiased, and so the mean of the estimates over many samples will approach the number of leaves.

In our case, this means that to produce an estimate of $|\Lambda(P)|$, we begin our estimate with the number of maximal elements in the poset, then delete one of the maximal elements uniformly at random.  We then multiply our estimate by the new number of maximal elements, and again delete one of the maximal elements uniformly at random.  We continue in this fashion until the poset is empty.

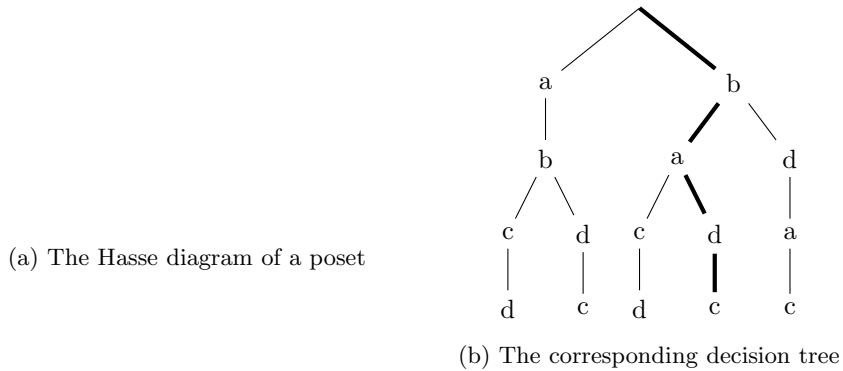
\begin{figure}[h]
\begin{subfigure}{0.5\textwidth}
\centering
\begin{tikzpicture}
\tikz \graph [no placement, nodes = {circle,draw}]{
a[x=0,y=1] -- c[x=1,y=0];
b[x=2,y=1] -- c;
b -- d[x=3,y=0];
};
\end{tikzpicture}
\subcaption{The Hasse diagram of a poset}
\end{subfigure}
\begin{subfigure}{0.5\textwidth}
\centering
\begin{tikzpicture}[
level distance = 10mm,
level 1/.style={sibling distance=25mm},
level 2/.style={sibling distance=15mm},
level 3/.style={sibling distance=10mm}]
\coordinate
child {
	node {a}
	child {
		node {b}
		child {node {c} child {node {d}}}
		child {node {d} child {node {c}}}
		}
	}
child[ultra thick] {
	node {b}
	child {
		node {a}
		child[thin] {node {c} child {node {d}}}
		child {node {d} child {node {c}}}
		}
	child[thin] {
		node {d}
		child {node {a} child {node {c}}}
		}
	};
\end{tikzpicture}
\subcaption{The corresponding decision tree}
\end{subfigure}
\caption{The highlighted path gives an estimate of $2\cdot 2\cdot 2\cdot 1=8$ in the classic Knuth algorithm.}
\end{figure}

One of the chief difficulties with Knuth's algorithm is that if the decision tree is far from uniform, the variance of the estimates may be large, and an exponential number of samples will be required for an accurate estimate.  Unfortunately, the degree of non-uniformity is unknown in general.  One way to mitigate the problem of non-uniform decision trees is to choose paths through the tree non-uniformly, according to an importance function.  This is known as \textit{sequential importance sampling} (SIS).

Consider a function $r$ from the elements of a poset $P$ to the positive real numbers.  We will refer to $r$ as the \emph{importance function} and to $r(v)$ as the \emph{importance} of $v\in P$.  By the importance of a set $S\subseteq P$, denoted $r(S)$, we mean the sum over $S$ of the importance of its elements.

To obtain an estimate of $|\Lambda(P)|$ using SIS, we again as before traverse a single path from the root to a leaf of the decision tree.  However, at each decision point, instead of choosing from among the maximal elements uniformly at random, we instead choose a maximal element with probability proportional to its importance.

To be more precise, if $v$ is a maximal element and $S$ is the set of all currently maximal elements, we choose $v$ from $S$ with probability $p(v) = r(v)/r(S)$.  The factor multiplied into our estimate at that branch point is then $1/p(v)$ rather than $|S|$, and our estimate of $|\Lambda(P)|$ is the product over each $v$ chosen of $1/p(v)=r(S)/r(v)$.  

For example, in the decision tree shown in figure 1, the highlighted path gives an estimate of
\[\left(\frac{r(a)+r(b)}{r(b)}\right)\left(\frac{r(a)+r(d)}{r(a)}\right)\left(\frac{r(c)+r(d)}{r(d)}\right)\left(\frac{r(c)}{r(c)}\right)\]

Because the probability of obtaining any given linear extension is the reciprocal of the estimate associated with that linear extension, when we take the expected value of the estimates, each linear extension contributes exactly 1 to the sum, and so the expected value of the estimates is $|\Lambda(P)|$.  Hence SIS gives an unbiased estimate for $|\Lambda(P)|$.

Pseudocode for finding a single estimate using this method appears in Algorithm \ref{SEA}.  For a poset with $n$ elements, the outer loop runs exactly $n-1$ times, while the inner loops run at most $n$ times.  Finding maximal elements can also be accomplished in linear time.  Hence the complexity of the algorithm is $O(n^2)$.

\begin{algorithm}[H]
\DontPrintSemicolon
\SetKwInOut{Input}{Input}\SetKwInOut{Output}{Output}
\Input{Poset $P$ and importance function $r:P\to \mathbb{R}_{>0}$}
\Output{An estimate of the number of linear extensions of $P$}
\BlankLine
$est \gets 1$\;
\While{$|P|>1$}
{
	$maxes\gets$ maximal elements of $P$\;
	$maxrsum\gets 0$\;
	\ForEach{$max\in maxes$}
	{
			$maxrsum\gets maxrsum+r(max)$\;
	}
	$randnum\gets$ random real number between $0$ and $maxrsum$\;
	$chosenmax\gets$ first element of $maxes$\;
	$currentrsum\gets r(chosenmax)$\;
	\While{$currentrsum<randnum$}
	{
		$chosenmax\gets$ next element of $maxes$\;
		$currentrsum\gets currentrsum+r(chosenmax)$\;		
	}
	$est\gets est\cdot maxrsum/r(chosenmax)$\;
	Delete element $chosenmax$ from $P$\;
}
	\Return est\;
\caption{Single Estimate Algorithm}\label{SEA}
\end{algorithm}

\section{Variance}  \label{Variance}
Here we calculate the relative variance of our estimates.  Note that by ``relative variance'' we mean the ratio of the variance to the square of the mean.  This is also the square of the coefficient of variation.  We denote this quantity $\mathrm{RV}(P)$ for a poset $P$.

The relative variance of our algorithm is given explicitly by the following theorem.

\begin{theorem}\label{rvexplicit}
For a poset $P$ with $L$ linear extensions, the relative variance of the estimates given by Algorithm \ref{SEA} is
\[\mathrm{RV}(P)=\left\langle\frac{f_P(\lambda)}{L}\right\rangle_u-1\]
where $f_P(\lambda)$ is the estimate associated with linear extension $\lambda$ and $\langle\cdot\rangle_u$ denotes the mean of the uniform distribution over all linear extensions.
\end{theorem}
\begin{proof}
Let $\Lambda(P)$ be the set of all linear extensions of poset $P$, let $f_P(\lambda)$ be the SIS estimate associated with extension $\lambda\in \Lambda(P)$, and let $L=|\Lambda(P)|$.  Note that $L$ is also the expected value of $f_P(\lambda)$.

Recall that the probability of selecting extension $\lambda$ is precisely $1/f_P(\lambda)$.  Hence
\[\mathbb{E}[f_P(\lambda)^2]=\sum_{\lambda\in \Lambda(P)}{f_P(\lambda)^2\frac{1}{f_P(\lambda)}}=\sum_{\lambda\in \Lambda(P)}{f_P(\lambda)}\]

Then, using our definition of relative variance, we have
\begin{equation*}
\begin{split}
\mathrm{RV}(P)&=\frac{\mathbb{E}[f_P(\lambda)^2]-\mathbb{E}[f_P(\lambda)]^2}{\mathbb{E}[f_P(\lambda)]^2}\\
&=\frac{\mathbb{E}[f_P(\lambda)^2]}{L^2}-1\\
&=\frac{1}{L^2}\left(\sum_{\lambda\in \Lambda(P)}{f_P(\lambda)}\right)-1\\
&=\frac{1}{L}\left(\sum_{\lambda\in \Lambda(P)}{\frac{f_P(\lambda)}{L}}\right)-1\\
&=\left\langle\frac{f_P(\lambda)}{L}\right\rangle_u-1
\end{split}
\end{equation*}
\end{proof}

We can also express the relative variance as a recursive function of the relative variance of certain sub-posets, as detailed in the following theorem.

\begin{theorem}\label{rvrecurthm}
For a poset $P$, let $M$ be the set of maximal elements of $P$, where for $m\in M$, $L_m$ is the number of linear extensions of $P$ that begin with $m$.  Then the relative variance of the estimates produced by Algorithm \ref{SEA} using importance function $r$ is given by
\[\mathrm{RV}(P)+1=\frac{r(M)}{L^2}\sum_{m\in M}{\frac{L_m^2}{r(m)}\big(\mathrm{RV}(P\setminus m)+1\big)}\]
\end{theorem}
\begin{proof}
First, observe that the linear extensions of $P$ that begin with a particular $m\in M$ are equivalent to the extensions in $\Lambda(P\setminus m)$, but with $m$ appended to the beginning of each extension in $\Lambda(P\setminus m)$.

Furthermore, $\lambda$ is an extension in $\Lambda(P)$ that begins with $m$, then the estimate produced by Algorithm \ref{SEA} for $\lambda$ in $P$ is related to the estimate for the equivalent extension, call it $\lambda'$, in $\Lambda(P\setminus m)$, by the following formula.
\begin{equation}\label{frecur}
f_P(\lambda) = \frac{r(M)}{r(m)}f_{P\setminus m}(\lambda')
\end{equation}

Note that every extension in $\Lambda(P)$ begins with some element $m\in M$, and so for each $\lambda\in\Lambda(P)$, there is $m\in M$ such that equation (\ref{frecur}) holds.

In the proof of Theorem \ref{rvexplicit} we showed that
\begin{equation}\label{rvp1explicit}
\mathrm{RV}(P)+1=\frac{1}{L^2}\sum_{\lambda\in \Lambda(P)}{f_P(\lambda)}
\end{equation}

Then, by substituting equation (\ref{frecur}) into equation (\ref{rvp1explicit}) and manipulating the summation expressions, we get

\begin{equation}\label{rvp1recur)}
\begin{split}
\mathrm{RV}(P)+1 &=\frac{1}{L^2}\sum_{m\in M}{\sum_{\lambda\in \Lambda(P\setminus m)}{\frac{r(M)}{r(m)}f_{P\setminus m}(\lambda)}} \\
&= \frac{r(M)}{L^2}\sum_{m\in M}{\frac{1}{r(m)}\sum_{\lambda\in \Lambda(P\setminus m)}{f_{P\setminus m}(\lambda)}} \\
&= \frac{r(M)}{L^2}\sum_{m\in M}{\frac{L_m^2}{r(m)}\frac{1}{L_m^2}\sum_{\lambda\in \Lambda(P\setminus m)}{f_{P\setminus m}(\lambda)}} \\
\end{split}
\end{equation}

Theorem \ref{rvexplicit} tells us that
\[\mathrm{RV}(P\setminus m)+1=\frac{1}{L_m^2}\sum_{\lambda\in \Lambda(P\setminus m)}{f_{P\setminus m}(\lambda)}\]

Hence, from the last line of equation (\ref{rvp1recur)}), we have
\[\mathrm{RV}(P)+1 =\frac{r(M)}{L^2}\sum_{m\in M}{\frac{L_m^2}{r(m)}\big(\mathrm{RV}(P\setminus m)+1\big)}\]
\end{proof}

This recursive form of the relative variance suggests a way to bound it, which is described in our next theorem.

\begin{theorem}\label{relvarbound}
Let $\mathcal{P}_i$ be the set of all posets of size $i$.  Let $A_i$ be given by
\[A_i:=\max_{P\in\mathcal{P}_i}\left(\max_{m\in M}\frac{r(M)}{r(m)}\frac{L_m}{L}\right)\]
where $M$, $L$, and $L_m$ are defined as before with respect to $P$.  Then for all posets $P$ of size $n$, the relative variance of the estimates produced by Algorithm \ref{SEA} using importance function $r$ is bounded by
\[\mathrm{RV}(P)\leq A_1 A_2\cdots A_n-1\]
\end{theorem}
\begin{proof}
The proof is by induction on the size of the poset $P$.

For the basis step, consider a poset of size 1.  The single element, call it $m$, is also the only maximum, and so $r(m)=r(M)$ and $L_m=L$.  Thus $A_1=1$.  Since this poset has only one possible linear extension, the algorithm always produces the same estimate, making the relative variance zero.  Thus the theorem statement is true for $n=1$.

Assume the theorem statement is true for all posets of size $n-1$ and less.  Then for a poset $P$ of size $n$, the theorem statement holds for the sub-posets $P\setminus m$, where $m\in M$.  That means
\[\mathrm{RV}(P\setminus m)+1\leq A_1 A_2\cdots A_{n-1}\]

Then, beginning with Theorem \ref{rvrecurthm}, we have
\begin{equation*}
\begin{split}
\mathrm{RV}(P)+1 &=\frac{r(M)}{L^2}\sum_{m\in M}{\frac{L_m^2}{r(m)}\big(\mathrm{RV}(P\setminus m)+1\big)} \\
&\leq \frac{r(M)}{L^2}\sum_{m\in M}{\frac{L_m^2}{r(m)}A_1 A_2\cdots A_{n-1}} \\
&= \sum_{m\in M}{\frac{L_m}{L}\frac{r(M)}{r(m)}\frac{L_m}{L}A_1 A_2\cdots A_{n-1}} \\
&\leq \sum_{m\in M}{\frac{L_m}{L}\max_{P\in\mathcal{P}_i}\left(\max_{m\in M}\frac{r(M)}{r(m)}\frac{L_m}{L}\right)A_1 A_2\cdots A_{n-1}} \\
&= \sum_{m\in M}{\frac{L_m}{L}A_1 A_2\cdots A_{n}} \\
&= A_1 A_2\cdots A_{n}\sum_{m\in M}{\frac{L_m}{L}} \\
\end{split}
\end{equation*}

Since every extension of $P$ must start with some $m\in M$,
\[\sum_{m\in M}{\frac{L_m}{L}}=1\]

Hence,
\[\mathrm{RV}(P)+1\leq A_1 A_2\cdots A_n\]
\end{proof}

From Theorem \ref{relvarbound}, we can see that the ideal importance function is $r(m)=L_m$, that is, the importance function which maps each maximal element to the number of extensions that begin with that element.  This importance function would give us $A_i=1$ for all $i\in\mathbb{N}$, making $\mathrm{RV}(P)=0$ for all posets $P$.

Of course, this ideal importance function is impossible to achieve, since if we knew its values, we would not need an estimation algorithm.  In Section \ref{ImportanceFunction}, we discuss more practical importance functions.

\section{The Importance Function}  \label{ImportanceFunction}

If the importance function $r$ is uniform across all elements at each decision point, then our algorithm is simply Knuth's algorithm again.  On the other hand, if a good importance function is available, one that reflects the actual structure of the decision tree, excellent results can sometimes be obtained, as in \cite{jcp}, \cite{algo}.  Here we present two importance functions, each with its own strengths and weaknesses, which we will discuss.

\subsection{The Number of Descendants}

For the first importance function, we define $r(v)$ to be $d(v)$, the number of descendants of $v$, including $v$ itself.  An example is shown in Figure \ref{impfcnexample}.

\begin{figure}
\begin{subfigure}{0.5\textwidth}
\centering
\begin{tikzpicture}
\tikz \graph [no placement, nodes = {circle,draw}]{
a[x=0,y=1,label=$2$] -- d[x=1,y=0];
b[x=2,y=1,label=$3$] -- d;
b -- e[x=3,y=0];
c[x=4,y=1,label=$1$];
};
\end{tikzpicture}
\subcaption{Number of descendants}
\end{subfigure}
\begin{subfigure}{0.5\textwidth}
\centering
\begin{tikzpicture}
\tikz \graph [no placement, nodes = {circle,draw}]{
a[x=0,y=1,label=$\frac{5+1}{5-1}$] -- d[x=1,y=0];
b[x=2,y=1,label=$\frac{5+2}{5-2}$] -- d;
b -- e[x=3,y=0];
c[x=4,y=1,label=$\frac{5+0}{5-0}$];
};
\end{tikzpicture}
\subcaption{Available spaces quotient}
\end{subfigure}
\caption{Example poset with maximal elements labeled with importance function values.}
\label{impfcnexample}
\end{figure}

The idea for this importance function came from the observation that each vertex must appear in a linear extension before all of its descendants, which themselves have fewer descendants.  In fact, if we order the vertices by decreasing number of descendants, we obtain a linear extension.

We also obtain a useful lower bound from this importance function.  The sum over the maximal elements of the number of their descendants must always be at least the number of vertices which have not been chosen yet, hence a lower bound for our samples and so for $|\Lambda(P)|$ itself is

\[|\Lambda(P)|\geq\frac{n!}{d(v_1)d(v_2)\cdots d(v_n)}\]

Here the denominator is the product over all $v_i\in P$ and is simply the observation that every $v_i$ is processed exactly once when generating a sample.

If the poset can be represented by a DAG that is a forest, then no maximal elements share descendants, and so every sample is equal to this lower bound, which is therefore exact.  In fact, these observations serve as an alternative proof of the formula for the number of linear extensions of a forest.  We also note in passing that this gives a way to sample from linear extensions of a forest uniformly at random.

In addition to a lower bound, this importance function also gives us an upper bound for each sample.  Consider that when constructing a linear extension, each element of the poset becomes maximal exactly once.  An element $v$ which was not maximal in the original poset $P$ becomes maximal when its last remaining ancestor and the ancestor's edge to $v$ are deleted.  If we collect this edge for each element of the poset, we have a spanning forest, $F$, of the poset $P$.

The number of linear extensions of this forest is clearly an upper bound for $|\Lambda(P)|$, and can be calculated exactly as

\[|\Lambda(F)|=\frac{n!}{d_F(v_1)d_F(v_2)\cdots d_F(v_n)}\]

where $d_F(v)$ is the number of descendants of $v$ in $F$.  Hence this upper bound depends on which forest we construct.

\subsection{The Available Spaces Quotient}

The second importance function we will discuss is

\[r(v) = \frac{i+d(v)-1}{i-d(v)+1}\]

where $i$ is the number of elements that have not yet been added to the linear extension.

Unlike the first importance function, this importance function changes depending on how far along in constructing the linear extension we are.  This clearly adds to the running time of the algorithm, but it does not increase the order of its complexity.

The motivation for this function came from the intuition that in addition to favoring elements with more descendants, we should also favor elements with fewer available spaces left.  At any point during the construction of a linear extension, if there are $i$ elements left to be added to the extension, then the number of available spaces for element $v$ is $i-d(v)+1$, since $v$ must be chosen before all of its proper descendants.

By placing the formula for the number of available spaces in the denominator, the importance of an element increases dramatically as the number of spaces available for that element decreases.  The particular formula for the numerator, on the other hand, was chosen as a result of numerical experimentation.

This importance function, unlike the first, does not give exact estimates for any known set of posets, nor has it given us any bounds on the number of linear extensions.  Its sole virtue is that it significantly reduces the observed variance of estimates in numerical testing.  See Section \ref{NumericalResults} for numerical results.

\section{Improvements on the Algorithm}  \label{Improvements}

Here we discuss a method of reducing the variance of our estimates by modifying the way in which we construct linear extensions.

Consider a DAG $D$ that, when undirected, consists of several connected components $D_1, D_2, ..., D_k$.  Each linear extension of $D$ can be obtained by first partitioning the positions $1, 2, ..., n$ into sets of size $|D_1|, |D_2|, ..., |D_n|$, and then filling the positions selected for each component with a linear extension of that component. Hence

\[ |\Lambda(D)| = \frac{|D|!}{|D_1|!|D_2|!\cdots |D_k|!}|\Lambda(D_1)|\cdot |\Lambda(D_2)|\cdots |\Lambda(D_k)| \]

This suggests a recursive algorithm whose pseudocode appears in Algorithm \ref{CC}.
\\
\\
\begin{algorithm}[H]
\DontPrintSemicolon

\SetKwFunction{fcn}{RecursiveEstimate}

\SetKwInOut{Input}{Input}\SetKwInOut{Output}{Output}
\fcn{$P$}\;
\Input{Poset $P$ and importance function $r:P\to \mathbb{R}_{>0}$}
\Output{An estimate of the number of linear extensions of $P$}

\BlankLine
\eIf{$|P|=1$} {$estimate \leftarrow 1$}
{
	Find connected components, $P_1, P_2, ..., P_k$ of $P$\;
	$estimate \leftarrow |P|!/\left(|P_1|!|P_2|!\cdots |P_k|!\right)$\;
	\For{$i=1$ to $k$}
	{
		$maxes\gets$ maximal elements of $P_i$\;
		$maxrsum\gets 0$\;
		\ForEach{$max\in maxes$}
		{
				$maxrsum\gets maxrsum+r(max)$\;
		}
		$randnum\leftarrow$ random real number between $0$ and $maxrsum$\;
		$chosenmax\leftarrow$ first element of $maxes$\;
		$currentrsum\leftarrow r(chosenmax)$\;
		\While{$currentrsum<randnum$}
		{
			$chosenmax\leftarrow$ next element of $maxes$\;
			$currentrsum\leftarrow currentrsum+r(chosenmax)$\;
		}
		$estimate\leftarrow estimate\cdot maxrsum/r(chosenmax)$\;
		Delete $chosenmax$ from $P_i$\;
		$estimate\leftarrow estimate\cdot$\fcn{$P_i$}\;
	}
}
\caption{Recursive Connected Components Algorithm}\label{CC}
\end{algorithm}

The order of the time complexity for Algorithm \ref{CC} is, at worst, $O(n^3)$.  This is because we have added the additional step of searching for connected components, which executes $n-1$ times and in the worst case takes quadratic time in the size of the poset being searched.  See Section \ref{NumericalResults} for numerical results on the variance of Algorithm \ref{CC}.

\section{Numerical Results}  \label{NumericalResults}

Numerical tests of Algorithm 1 were implemented in C++ using a sparse representation of the posets.  Given the poset elements $v_1,v_2,\dots,v_n$, for each pair of elements $v_i$ and $v_j$ with $i<j$, the relation $v_i>v_j$ was given a 20\% probability to exist using a pseudorandom number generator.  The posets were then transitively completed.

For each value of $n$ from $10,15,\dots,150$, we generated $n^2$ posets in this manner, and $n^2$ SIS estimates were performed on each poset.  The relative variance of the estimates was calculated for each poset, and the results averaged for each value of $n$.

The results for the different importance functions were of differing orders of magnitude; therefore, they are compared in Figure \ref{numresults1} on a log-log scale.  Linearity on a log-log scale suggests a power function relationship where the exponent is the slope of the line.

\begin{figure}[H]
\includegraphics[width=\textwidth]{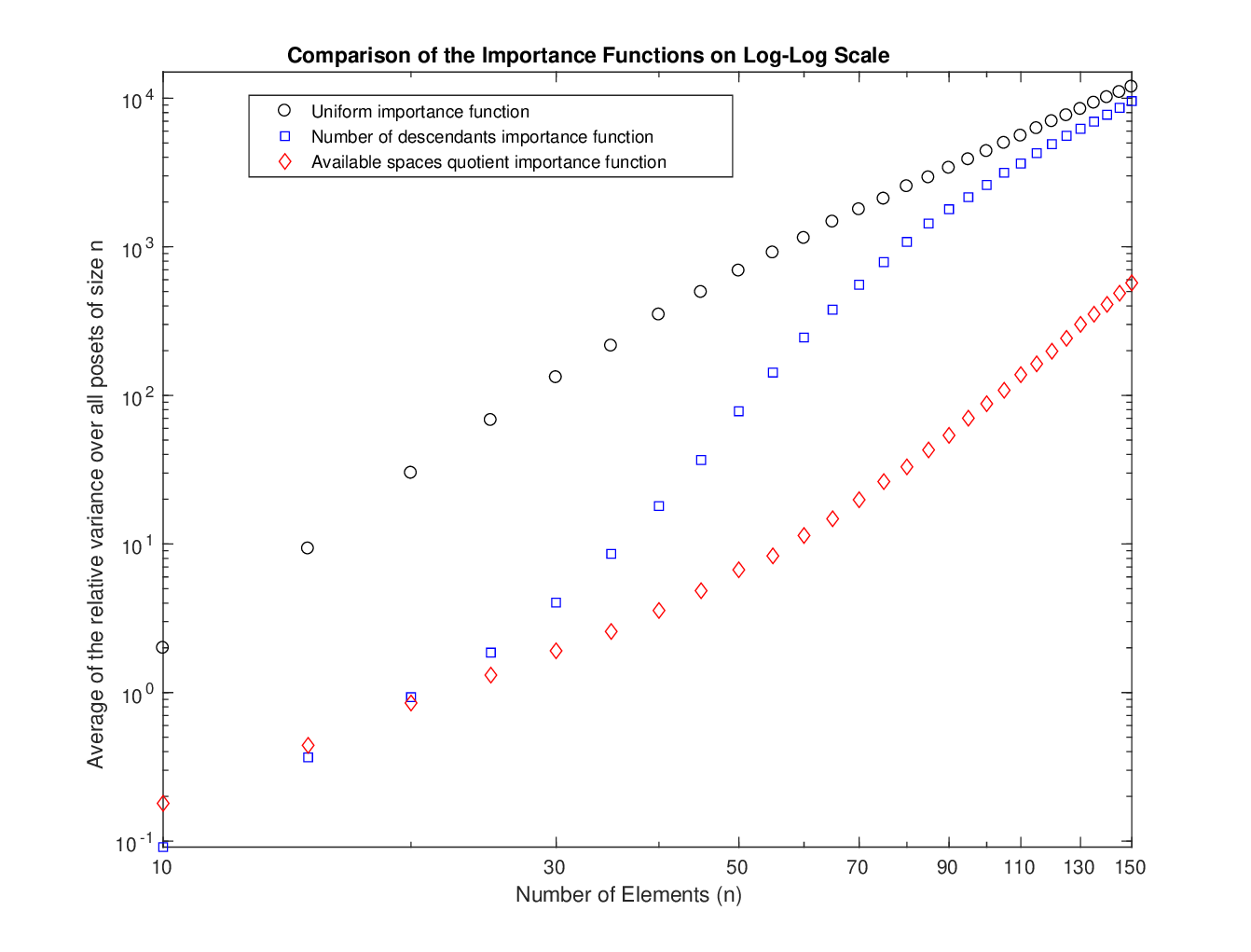}
\caption{Comparison of the average relative variance of the three importance functions, on a log-log scale.}\label{numresults1}
\end{figure}

Numerical tests of Algorithm 2, the Recursive Connected Components algorithm, were performed in the same manner as for Algorithm 1.  The results in Figure \ref{numresults2} show a modest reduction in the average relative variance for each of the three importance functions when recursion is used.  The relative size of this reduction appears to diminish as the poset size increases.

\begin{figure}[H]
\includegraphics[width=\textwidth]{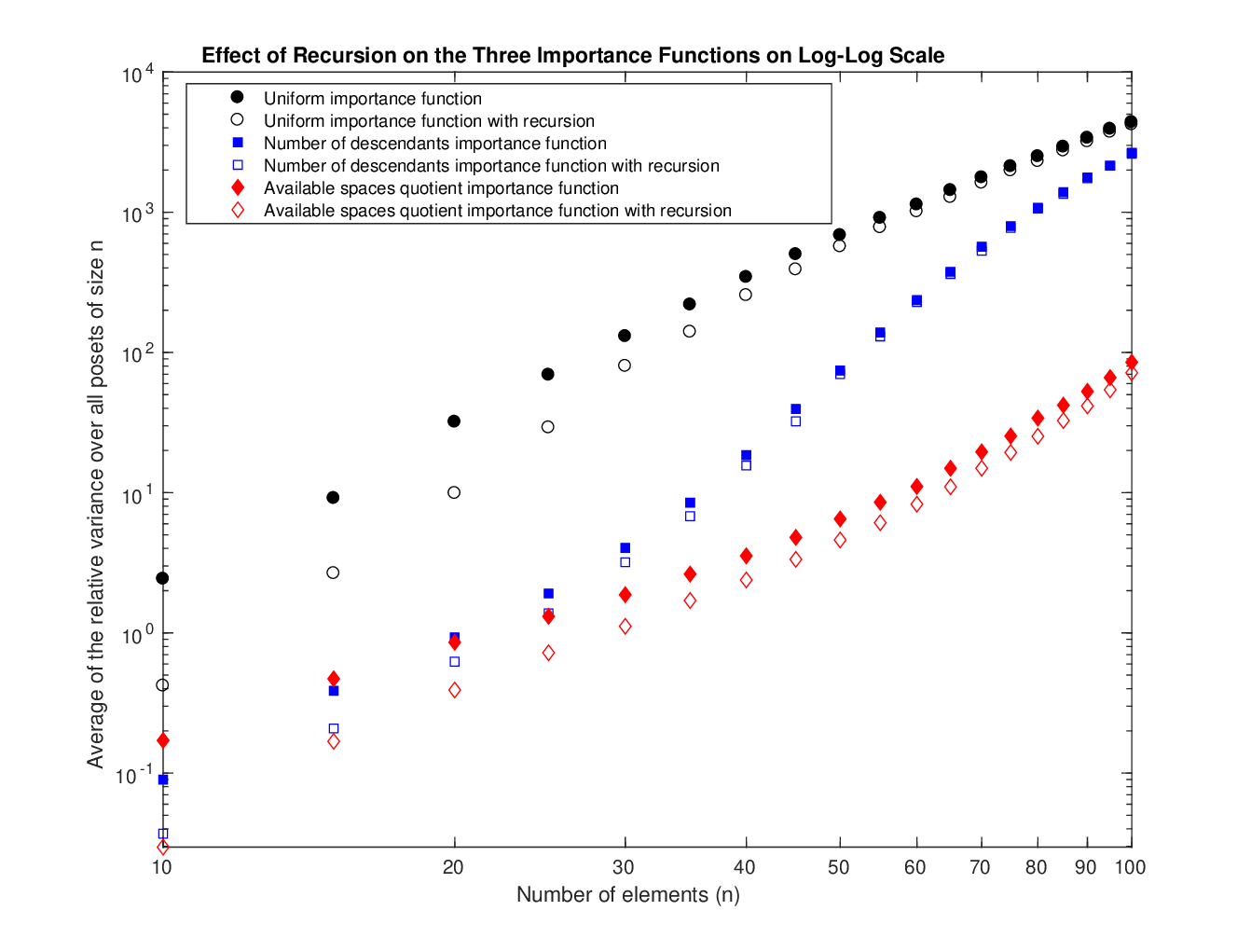}
\caption{Comparison of the Recursive Connected Components algorithm with the original algorithm for the three importance functions, on a log-log scale.}\label{numresults2}
\end{figure}

These numerical results show that the variance of our method does not grow too quickly, and that our importance functions significantly reduce the variance as compared to a uniform importance function.

\section*{Acknowledgements}
We thank Myra Deng and Amanda Crawford for work done on early versions of the algorithm.

\bibliographystyle{amsplain}
\bibliography{lxbib}

\end{document}